\documentclass[final,journal]{IEEEtran}

\usepackage{graphicx}
\usepackage{color}
\graphicspath{{Figures/}}
\usepackage{amsmath,amsthm,amssymb}
\usepackage{listings}

\newtheorem{lemma}{Lemma}

\newtheorem{theorem}{Theorem}
\newtheorem{definition}{Definition}

\lstnewenvironment{matlab}%
{\lstset{mathescape=true,
numbers=left,numberstyle=\tiny,numbersep=5pt,columns=flexible,%
keywords={for,end,function,if,else,while},%
moredelim=[is][\ttfamily]{|}{|}}}{}

\title{Graph-Theoretic Analysis of Residual Generation Under Computational Constraints}
\author{Jan {\AA}slund%
\thanks{Department of Electrical Engineering,
Linköping University, Linköping, Sweden 
(e-mail: jan.aslund@liu.se).}
}

\date{}
\begin{document}
  \maketitle

\begin{abstract}
A unified structural framework is presented for model-based fault diagnosis 
that explicitly incorporates both fault locations and constraints imposed by the 
residual generation methodology. Building on the concepts of proper and minimal 
structurally overdetermined (PSO/MSO) sets and Test Equation Supports (TES/MTES), 
the framework introduces testable PSO sets, Residual Generation (RG) sets, irreducible
fault signatures (IFS),
and Irreducible RG (IRG) sets to characterize which submodels are suitable for 
residual generation under given computational restrictions. 
An operator $M^*$ is defined to extract, from any model, 
the largest testable PSO subset consistent with a specified residual generation method. 
Using this operator, an algorithm is developed to compute all RG sets, 
and it is shown that irreducible fault signature sets form the join-irreducible elements of a join-semilattice 
of sets and fully capture the multiple-fault isolability properties 
in the method-constrained setting. The approach is exemplified on a 
semi-explicit linear DAE model, where low structural differential index
can be used to define $M^*$. The results demonstrate that the proposed framework generalizes 
 MTES-based analysis to residual generation scenarios with explicit computational limitations.

\end{abstract}


   \section{Introduction}

 A well-established approach in model-based diagnosis is to exploit 
 redundancy within a system model to design analytical tests capable 
 of detecting inconsistencies. Numerous methods have been proposed to identify 
structurally redundant submodels suitable for this purpose. 
In~\cite{MSOalg:2008,Gelso:2008}, algorithms were developed for 
identifying Minimal Structurally Overdetermined (MSO) sets. 
The method of Possible Conflicts was introduced in~\cite{Pulido:2004}, 
while~\cite{Trave:2006} employed Structural Analytical Redundancy Relations. 
A comprehensive overview of various definitions and algorithms 
can be found in~\cite{SAFFE:09:Armengol}.

In addition to redundancy, an important aspect is the methodology employed 
for test construction. Several studies have addressed this issue by 
proposing structural characterizations of equation subsets that 
are applicable under specific constraints imposed by the residual 
generation process.
In~\cite{SMCA:10:CS:MN:2010}, a method for identifying MSO sets 
for sequential residual generation was proposed, 
explicitly considering the computational tools 
available for solving algebraic and differential
equations as well as for numerical differentiation. 

The notion of causally computable MSO sets was later 
introduced in~\cite{SMCA11:Causal:}, where computational
sequences were derived under the assumption 
that linear square subsystems are invertible, 
subject to causality constraints among variables. 
In~\cite{SMCA12:CausalIsol}, the concept of the monitorable part was defined under 
derivative, integral, and mixed causality assumptions, and the 
implications for fault isolability were systematically analyzed.
Furthermore,~\cite{AAC:VNLE:2016} presented a method for identifying
diagnosable subsystems characterized by their suitability for observer-based 
residual generation and their ability to achieve maximal fault isolability. 

The degree of redundancy also plays a critical role in the design of 
diagnostic tests. For instance, in observer-based approaches, high redundancy 
is desirable as it enhances robustness to modeling errors and measurement noise. 
However, increased redundancy may reduce fault isolability, 
since residuals derived from larger submodels tend to be sensitive to a broader set of faults. 
A structured approach to address this trade-off was introduced 
in~\cite{DX10MFS:2010}, where the concept of Test Equation Supports (TESs) 
was defined to account for both the level of redundancy
and the distribution of faults in the system model.

The aim of this work is to integrate the aforementioned perspectives 
through a unified framework that considers both fault locations and 
the residual generation methodology. The analysis further investigates 
the implications of this integration and establishes 
an appropriate graph-theoretical basis.

  \section{Theoretical background}
\label{sec:Theory}

Graph theory provides a useful framework for analyzing structural properties
 of systems and for identifying submodels with redundancy. 
 Consider a set of equations $M = \{ e_1, \ldots, e_m \}$ defined over a set of unknown 
 variables $X = \{ x_1, \ldots, x_n \}$. A bipartite graph can then be constructed with 
 node sets $M$ and $X$, and an edge set $E \subset M \times X$, 
 where $(e_i, x_j) \in E$ if the variable $x_j \in X$ appears in equation $e_i \in M$.

As an illustrative example, consider the system of ordinary differential 
equations (ODEs)
\begin{align*}
  \dot{x}_{1} &= x_{1} + x_{3},\\
  \dot{x}_{2} &= x_{2} - x_{3},\\
  \dot{x}_{3} &= x_{1}.
\end{align*}
The corresponding bipartite graph can be represented by the bi-adjacency matrix
\begin{equation*}
  \quad\begin{bmatrix}
    X&&X\\
    &X&X\\
    X&&X\\
  \end{bmatrix}
\end{equation*}
where an $X$ in position $(i,j)$ indicates that $(e_i, x_j) \in E$. 
An alternative representation that explicitly incorporates differential 
constraints is described in~\cite{blanke2016diagnosis}.

The Dulmage--Mendelsohn (DM) decomposition of bipartite graphs is a powerful 
tool for structural analysis~\cite{blanke2016diagnosis}, \cite{Mendelsohn:1958}. 
The decomposition, illustrated for a general system in Figure~\ref{fig:dmperm}, 
permutes the rows and columns of the bi-adjacency matrix so that three distinct 
parts of the model can be identified: an overdetermined part $M^+$ with more equations 
than unknowns ($|M^+| > |X^+|$), an exactly determined part $M^0$,
which introduces as many new variables $X^0$ as equations ($|M^0| = |X^0|$), 
and an underdetermined part $M^-$, which introduces more new variables $X^-$ than 
equations ($|M^-| < |X^-|$).

\begin{figure}[htbp]
    \centering
    \includegraphics[width=0.3\textwidth]{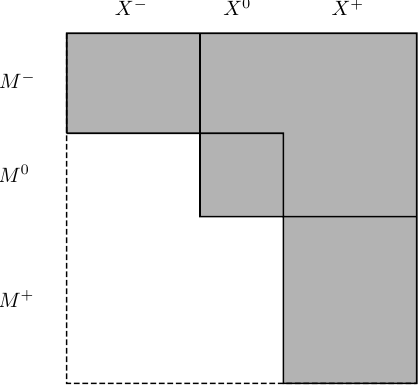}
    \caption{Illustration of the bi-adjacency matrix in a Dulmage--Mendelsohn decomposition, 
    where all edges lie within the shaded area}
    \label{fig:dmperm}
\end{figure}

The overdetermined part $M^+$ is of particular interest for diagnosis, since it contains 
redundant information that can be used to detect inconsistencies. 

A characterization that a matrix is {\it structurally invertible}
means that almost all numerical matrices with the same
zero–nonzero pattern are invertible. 
Similar structural characterizations will be used later in the paper.
The following definitions from~\cite{MSOalg:2008} provide two structural characterizations 
of redundant subsets that are suitable for the design of diagnostic tests.
\begin{definition}[PSO set, MSO set, and redundancy $\varphi$]\label{def:msopso}
  A set of equations $M$ is Proper Structurally Overdetermined (PSO) 
  if $M^+ = M$. A PSO set is Minimal Structurally Overdetermined (MSO) 
  if no proper subset of $M$ is PSO. The redundancy $\varphi(M)$ of a PSO set $M$ is defined
  as the difference between the number of equations and the number of unknown variables in $M$.
\end{definition}

In~\cite{DX10MFS:2010}, the location of faults was incorporated into the structural analysis 
through the following definitions.
\begin{definition}[TES and MTES]\label{def:tesmtes}
  Let $F(M)$ denote the set of faults included in an equation set $M$.
   A PSO set $M$ is a  Test Equation Support (TES) if $F(M) \neq \emptyset$ and 
   for any PSO set $M'\supsetneq M$ it holds that $F(M')\supsetneq F(M)$.
   A TES is minimal (MTES) if no proper subset of $M$ is a TES.
\end{definition}
In the definition of MTESs, fault locations in the model are explicitly taken 
into account. The main objective of this paper is to generalize these 
structural concepts by also incorporating the residual generation method 
into the analysis, and thereby to identify a graph-theoretical representation 
that captures both fault distribution and computational constraints.

For further examples, discussions, and information about implementation of 
MSO sets and MTESs algorithms, see \cite{glotin2024dx}  and \cite{FRISK20173287}.

  \section{Basic definitions}\label{sec:definitions}
The following system of linear equations will be used in the forthcoming
sections to illustrate some basic ideas and to highlight the consequences
of imposing restrictions on the method used in the design of diagnostic
tests.

  \begin{alignat}{3}
    e_1:&&\quad x_1+2x_2&=u_1+v_1\nonumber\\
    e_2:&&x_1+x_2&=u_2+v_2\nonumber\\
    e_3:&&y_1&=x_1+3x_2+v_3\label{eq:exsyst}\\
    e_4:&&y_2&=x_1-x_2+f_1+v_4\nonumber\\
    e_5:&&y_3&=x_2+f_2+v_5\nonumber
  \end{alignat}

 Here $u_i$ and $y_i$ are known signals, $x_i$ are unknown variables,
 and the vector ${\bf v}=(v_1,\ldots,v_5)^T$ is gaussian random vector
 with zero mean and identity covariance matrix.
 The fault signals $f_1$, and $f_2$, are modeled as additive faults,
 and note that it is not assumed that all equations
 include a fault. This is often the case for basic physical equations,
 or faults that can be discovered by other means and therefore can be
 handled separately. 
 One example of the latter is sensors with  a defined behavior under 
 fault conditions, such as a sensor that is considered healthy when its output remains within 
 a specified fail‑safe range 
 and indicates a fault when the signal leaves this range;
 see e.g. Chapter 12 and 13 in  \cite{dunn2005fundamentals}. 

\subsection{Restrictions on the residual generation methodology}
From now on, it is assumed that variables in system~\eqref{eq:exsyst}
can only be computed by sequential back-substitution, and that no row
operations to eliminate variables are allowed.

The purpose of analyzing system~\eqref{eq:exsyst} under this artificial
computational constraint is to mimic situations that may arise, for
example, when considering system models with non-invertible subsystems.
A more natural computational constraint will be presented in
Section~\ref{sec:examples}.

\subsection{Residual generation}
In this section, the possibility to create residual sensitive to faults in model 
\eqref{eq:exsyst} is investigated under the restriction introduced
above. First, the noise signals are omitted and it is assumed that the
system operates in the fault-free mode. Consider the variable $x_2$.
Under the computational restriction, the only way to compute $x_2$ is to
use equation $e_5$ in model~\eqref{eq:exsyst},
\begin{equation*}
  x_2 = y_3.
\end{equation*}
Then $x_1$ can be computed using, for example, equation $e_2$,
\begin{equation*}
  x_1 = u_2 - y_3.
\end{equation*}

By substituting $x_1$ and $x_2$ into the redundant equation $e_1$, the
residual
\begin{equation*}
  r_1 = u_1 - u_2 - y_3 = 0
\end{equation*}
is obtained, which is equal to zero in the fault-free case. Including
faults and noise in the model yields
\begin{equation*}
  r_1 = u_1 - u_2 - y_3 = f_2 + w_1,
\end{equation*}
where $w_1 = -v_1 + v_2 - v_5$. The first expression shows how the
residual is computed from the known signals $u_1$, $u_2$, and $y_3$,
while the second expression shows how the residual depends on the
unknown fault $f_2$ and the noise $w_1$, which can be used to evaluate
the fault-to-noise ratio of the residual.

In this case, the MSO set $\{e_1,e_2,e_5\}$ was used to construct the
residual $r_1$. This was possible since, first, the unknown variables in
the submodel could be computed under the imposed restriction and,
second, there was redundancy in the model, i.e., the system was
overdetermined. These two properties can be read directly from a graph representation
of the system and admit a straightforward structural characterization.

\begin{definition}[Structurally testable]\label{def:ss}
  Let a computational constraints be specified by a structural characterization
  of which equation sets can be used to generate residuals that are sensitive to all 
  faults included in the equation set.
  A PSO set that satisfies this structural characterization is called
  structurally testable.
\end{definition}

Next, it is shown how the fault-to-noise ratio can be improved by
considering a larger PSO set with more redundant equations. As a first
step, the MSO set $\{e_1,e_3,e_5\}$, instead of $\{e_1,e_2,e_5\}$, is
used to derive the residual
\begin{equation*}
  r_2 = y_1 + y_3 + u_1 = f_2 + w_2,
\end{equation*}
where $w_2 = v_1 + v_5 - v_3$. The residual has been scaled so that the
dependence on the fault $f_2$ is the same for both residuals. The next
step is to form the affine combination
\begin{equation*}
  r = k r_1 + (1-k) r_2 = f_2 + v,
\end{equation*}
where $v = k w_1 + (1-k) w_2$, and choose $k$ to minimize the variance
of $v$.

It is straightforward to show that the covariance matrix of the random
vector $\mathbf{w} = \begin{pmatrix} w_1 & w_2 \end{pmatrix}^T$ is
\begin{equation*}
  \Sigma_{\mathbf{w}} =
  \begin{pmatrix}
    \sigma^2_{11} & \sigma_{12}\\
    \sigma_{21}   & \sigma^2_{22}
  \end{pmatrix}
  =
  \begin{pmatrix}
    3 & -1\\
    -1 & 3
  \end{pmatrix}.
\end{equation*}
By minimizing the variance of $v$ with respect to $k$, the residual
\begin{align*}
  r
  &= \frac{\sigma^2_{22} - \sigma_{12}}{\sigma^2_{11} + \sigma^2_{22} - 2\sigma_{12}}\, r_1
   + \frac{\sigma^2_{11} - \sigma_{12}}{\sigma^2_{11} + \sigma^2_{22} - 2\sigma_{12}}\, r_2\\
  &= \frac{1}{2} r_1 + \frac{1}{2} r_2 = f_2 + v
\end{align*}
is obtained for $k = 1/2$, with variance
\begin{equation*}
  \sigma_v^2
  = k^2 \sigma_{11}^2
    + 2k(1-k)\sigma_{12}
    + (1-k)^2 \sigma_{22}^2
  = 1.
\end{equation*}
Comparing this value with the variance of the noise in the individual
residuals $r_1$ and $r_2$, i.e., $\sigma_{11}^2 = \sigma_{22}^2 = 3$,
shows that by using the larger testable PSO set
$M_{1} = \{e_1,e_2,e_3,e_5\}$ it is possible to construct a residual
with a higher fault-to-noise ratio than for either $r_1$ or $r_2$
alone. The technique used above to combine the two residuals is a standard example
from the theory of sensor fusion~\cite{Kay1993}.

Higher order of redundancy can be exploited in many dynamic and
nonlinear models to further suppress model and measurement noise when
designing residuals.
In the MSO algorithm in~\cite{MSOalg:2008} and the MTES algorithm in~\cite{DX10MFS:2010},
a crucial step is the computation of the overdetermined part $M^+$ 
defined as the largest subset of $M$ that is a PSO set. The operator $M^*$, defined below,
will replace the operator $M^+$ in the algorithm introduced in the Section~\ref{sec:algorithm}.

The definition of the operator $M^+$ was extened in~\cite{SMCA11:Causal:} to an operator 
that gives the largest PSO subset of a model
for which all variables can be computed under the assumption that linear square
subsystems can be inverted and subject to causality constraints between variables.
In~\cite{SMCA12:CausalIsol}, the monitorable part of a model was defined in an analogous
way under derivative, integral, and mixed causality assumptions.
The following general operator will be considered in this paper.
\begin{definition}[The operator $M^*$]
  Given a set of equations $M$, the set $M^*$ is defined as the largest
  testable PSO subset of $M$.
\end{definition}

\subsection{Fault signature}
Different residuals are sensitive to different sets of faults, and the
possible sets of fault signatures of the residuals are given by the
following definition.
\begin{definition}[Fault signature]\label{def:fset}
  Let $F(M)$ denote the set of faults included in a model $M$. A set of
  faults $F\neq\emptyset$ such that $F = F(M)$ for some testable PSO set $M$ is
  called a fault signature.
\end{definition}
The set $M_{1} = \{e_1,e_2,e_3,e_5\}$ used above has the fault
signature $F(M_1) = \{f_2\}$, and $M_1$ is the largest testable PSO set
with this fault signature. Hence, the fault-to-noise ratio cannot be
improved further by adding more equations without increasing  the set of
faults which the residual is sensitive to. This makes $M_1$ suitable for
constructing a residual sensitive to $f_2$.

The following lemma shows that for any given fault signature there is a
largest testable PSO set with that signature.
\begin{lemma}\label{lem:mset}
  Given a fault signature $F$, there exists a unique testable PSO set
  $M$ such that $F(M) = F$ and $M' \subset M$ for all testable PSO sets
  $M'$ satisfying $F(M') = F$.
\end{lemma}
\begin{proof}
  Let $M$ be the union of all testable PSO sets $M'$ such that
  $F(M') = F$. By Definition~\ref{def:fset}, there exists at least one
  testable PSO set $M'$ with $F(M') = F$, so the union is non-empty.
  Moreover, $M$ is a testable PSO set since it is a union of testable
  sets, and it clearly contains every such $M'$ as a subset.
\end{proof}

It follows that the sets given by Lemma~\ref{lem:mset} are of particular
interest when designing residuals, which motivates the following
definition.
\begin{definition}[RG set]\label{def:rset}
  A testable PSO set $M$ is called a Residual Generation set
  (RG set) if it is the largest testable PSO set such that $F(M) = F$
  for a given fault signature $F$.
\end{definition}

\subsection{Fault detectability and fault isolability}
The concepts structurally detectable faults and structurally isolability of multiple 
faults from \cite{DX10MFS:2010} 
are now modified by replacing the $M^+$ operator by $M^*$. It is assumed that a 
fault $f$ affects only one equation and that equation is denoted by $e_f$. 
\begin{definition}[Structurally detectable]
  Let $F(M)$ denote the set of faults included in a model $M$.
  A fault $f$ is structurally detectable in $M$ if $f\in F(M^*)$.
\end{definition}
A fault mode is represented by a set $F_i$ of faults, and a fault mode $F_i$ is isolable 
from a fault mode $F_j$ if there exists a residual sensitive to some fault
of $F_i$ but insensitive to all faults in $F_j$. A formal definition is:
\begin{definition}[Structurally isolable]
  A fault mode $F_i$ is structurally isolable from mode $F_j$ in a model $M$ if
\begin{equation*}
  \exists f\in F_i: e_f\in (M \setminus eq(F_j))^*,
\end{equation*}
where $eq(F_j)=\cup_{f\in F_j} \{e_f\}$.
\end{definition}

The set of all TESs  forms a partially ordered set under set 
inclusion, where the MTESs are the minimal elements in the graph-theoretical sense.
For each TES there is an associated set of faults included in the corresponding submodel
given by the following definition from \cite{DX10MFS:2010}.
\begin{definition}[TS and MTS]
A subset of faults $F$ is 
a test support TS if there exists a PSO set $M$ such that $F(M)=F$.
A test support is a minimal test support MTS if no proper submodel is a test support
\end{definition}
There is a one-to-one correspondence between the TESs and the TSs, 
and the collection of all TSs characterizes the complete multiple-fault isolability 
properties of the model.  
Hence, the MTESs and the corresponding MTSs contain all information about the multiple 
fault isolability properties of the model.

We now turn to the corresponding results for RG sets. 
In general, the set of all RG sets also forms a partially ordered set under inclusion. 
There is again a one-to-one correspondence between the RG sets and all possible 
fault signatures, which completely characterizes the multiple-fault isolability properties 
of the model under the restrictions considered in this paper. 
The main difference between TESs and RG sets can be illustrated using model~\eqref{eq:exsyst}. 
In this example, there are two RG sets, $M_1 = \{e_1,e_2,e_3,e_5\}$ and 
$M_2 = \{e_1,e_2,e_3,e_4,e_5\}$, with corresponding fault signatures $F_1 = \{f_2\}$ 
and $F_2 = \{f_1,f_2\}$, respectively. The set $M_1$ is the only minimal element. 

Consequently, it is not sufficient to consider only the minimal RG sets 
when analyzing the structural fault detection and isolation properties of the model.
A union of RG sets is a subset of an RG set with a fault signature that is equal to
the union of the fault signatures of the individual RG sets. Hence, the union of
RG sets does not provide any additional information about the fault isolability
properties of the model. 
This motivates the following definition.
\begin{definition}[Irreducible fault signature and IRG set]\label{def:iset}
A fault signature is called irreducible if it cannot be written as a union other fault signatures.
An RG set is called an Irreducible RG (IRG) set if its fault signature is irreducible.
\end{definition}

In graph-theoretical terms, the set of all fault signatures can be viewed as a join-semilattice, 
where the join operation is given by set union, and the irreducible fault signatures of
 the join-irreducible elements in this structure; see, for example,~\cite{MR0227053}.

The following result summarizes the discussion and shows that, 
in the context of this paper, IRG sets play a role analogous to that of MTESs.
\begin{theorem}\label{thm:i}
There is a one-to-one correspondence between the set of all RG sets and the
set of all fault signatures. Furthermore, the fault signature of any RG set can be written as 
as a union of irreducible fault signatures, and the set of all irreducible fault signatures
is the minimal set with this property.
\end{theorem}

\begin{proof}
  The one-to-one correspondence follows directly from Definition~\ref{def:rset} 
  and Lemma~\ref{lem:mset}. 
  By Definition~\ref{def:iset},  any fault signature is either an irreducible fault signatur itself,
   or  can be represented as a union of irreducible fault signatures. 
  Finally, any collection of sets with this property must contain all irreducible fault signatures,
  which proves minimality.
\end{proof}

  \section{Algorithm}\label{sec:algorithm}

The operator $M^{*}$ will be used in the algorithm below. Before presenting it, three
important differences compared to the MTES algorithm are highlighted, since they must be
taken into account when modifying the MTES algorithm.

The first difference concerns a key operation. In the MTES algorithm, a frequently used
operation is to remove one equation and then compute the overdetermined part,
$(M\setminus\{e\})^+$. If $e$ is an arbitrary equation in a PSO set $M$, then
\begin{equation*}
  \varphi\left((M\setminus\{e\})^+\right) = \varphi(M) - 1,
\end{equation*}
that is, the order of redundancy decreases by one. 

In the algorithm below, the
corresponding operation is to first remove one equation and then compute the largest testable
subset $(M\setminus\{e\})^*$. 
For system~\eqref{eq:exsyst}, the order of redundancy is
$\varphi(M)=3$. If equation $e_5$ is removed, the largest testable subset is empty,
$(M\setminus\{e_5\})^*=\emptyset$, and thus
\begin{equation*}
  \varphi\left((M\setminus\{e_5\})^*\right) = 0 < \varphi(M)-1 = 2.
\end{equation*}

The second difference is that, for a given model, all MTESs have the same order of
redundancy, which can be computed in advance; see Lemma~3 in~\cite{DX10MFS:2010}.
This property is used as a stopping criterion in the MTES algorithm. In contrast, the
IRG sets of model~\eqref{eq:exsyst} have different orders of redundancy:
$\varphi(\{e_1,e_2,e_3,e_5\}) = 2$ and $\varphi(\{e_1,e_2,e_3,e_4,e_5\}) = 3$.

The third difference is related to structural isolability as defined
in~\cite{Krysander:2008}, which is symmetric: if $f_i$ is isolable from $f_j$, then
$f_j$ is isolable from $f_i$.
In the example in Section~\ref{sec:definitions}, the two fault signatures are $\{f_2\}$ and
$\{f_1,f_2\}$. This means that $f_2$ is structurally isolable from $f_1$, but not vice
versa. 
Without this symmetry, it is not possible to define equivalence classes as in
Section~V.B of~\cite{MSOalg:2008}, which were used to improve the efficiency of the MSO
and MTES algorithms.
Furthermore, due to the lack of symmetry, it is not possible to define the sets $\mathcal{R}$ which 
were used in these algorithms to prevent that the same set was found more than once.

The algorithm below follows the same basic ideas as in~\cite{DX10MFS:2010}. 
It is initialized with the PSO sets $M = M_0^*$, where $M_0$ is the model to be analyzed.
Equations with faults are removed and largest testable subsets are computed alternately 
until the there is a set without faults reached in each branch of the recursion tree.

The output of the algorithm is the set of all RG subsets of $M_0$.
It is straightforward to identify which RG sets are IRG sets using
Definition~\ref{def:iset}, so this step is not included here.
\begin{matlab}
  function $\mathcal{S}$ = |FindRG|($M$)
    $\mathcal{S} = \emptyset$;
    while $F(R)\neq\emptyset$
      Select an $e\in R$ such that $F(\{e\})\neq\emptyset$
      $M'=(M\setminus\{e\})^*$;
      if $F(M')\neq\emptyset$
        $\mathcal{S}=\mathcal{S}\cup \{M'\};$
        $\mathcal{S}=\mathcal{S}\cup \verb+FindRG+(M');$
      end
    end
  end
\end{matlab}

One modification compared to~\cite{DX10MFS:2010} is that the operator $M^{+}$ has been
replaced by $M^{*}$ in line~5. 
Moreover, since the order of redundancy of the IRG sets is
not known a priori, the stopping criterion in line~6 is that $F(M')=F((M\setminus\{e\})^*)$ becomes empty, rather
than reaching a precomputed redundancy level. 
This situation occurs, for example, if $e_5$ is removed from system~\eqref{eq:exsyst}. 

The section concludes with a proof establishing the correctness of the algorithm.
\begin{theorem}
  The output of the algorithm is the set of all $RG$ sets in the model $M_0$.
\end{theorem}
\begin{proof}
Input to the algorithm is $M=M_0^*$.
We show that an arbitrary RG set $M' \subsetneq M$ is found at least once 
by the algorithm by constructing a branch in the recursion tree from $M$ to $M'$.

Assume that $f\in F(M)\setminus F(M')$, and define 
\begin{equation*}
  M^s=(M\setminus\{e_f\})^*.
\end{equation*}
We claim that $F(M') \subset F(M^s)$.  
If this were not the case, then
\begin{equation*}
  F(M^s) \subsetneq F(M'\cup M^s) \subset F(M) \setminus \{f\} = F(M \setminus \{e_f\}),
\end{equation*}
where $M'\cup M^s \subset M$ is a testable PSO set, which contradicts the definition of $M^s$.

If $M^s = M'$, the claim follows.  
Otherwise, we repeat the procedure above with $M$ replaced by $M^s$ until $M'$ is reached.
\end{proof}

\section{Example: Differential algebraic equations}
\label{sec:examples}
In this section, semi-explicit linear differential-algebraic equations (DAEs) of the form
\begin{subequations}\label{eq:SEDAE}
  \begin{align}
    \dot{\bf x}_1&=A_{11}  {\bf x}_1+A_{12}  {\bf x}_2+  B_{1} {\bf u},\label{eq:SEDAEa}\\
    0&=A_{21}  {\bf x}_1+A_{22}  {\bf x}_2+  B_{2}{\bf u},\label{eq:SEDAEb}\\
    {\bf y}&=C_1{\bf x}_1+C_2{\bf x}_2,\label{eq:SEDAEc}
  \end{align}
\end{subequations}
are be used to illustrate how the theory developed in the previous sections can be applied to
identify subsystems that are compatible with a specified residual generation methodology.
The following example is considered.
\begin{equation}\label{eq:exDAE}
  \begin{alignedat}{2}
    e_1:&& \quad \dot{x}_{1}&=x_{2}-x_{1}\\
    e_2:&&  \dot{x}_{2}&=x_{1}+x_{2}-2x_{3}+f_{1}\\
    e_3:&&  \dot{x}_{3}&=x_{2}-x_{3}\\
    e_4:&& y_{1}&=x_{1}+f_{2}\\
    e_5:&& y_{2}&=x_{3}+f_{3}\\
    e_6:&& y_{3}&=x_{1}-x_{3}
  \end{alignedat}
\end{equation} 
which is an ODE system. 
However, the analysis below considers subsystems of \eqref{eq:exDAE}, 
which are in general semi-explicit linear DAEs. 

In this example, a set of equations is defined to be structurally testable 
if it has low structural differential index. 
The differential index can be seen as a measure of how far a DAE is from being an ordinary
differential equation (ODE); an index of zero corresponds to an ODE system.
Models with index zero or one are referred to as low-index DAEs and can be integrated
using standard ODE solvers; see, e.g.,~\cite{ascher1998computer}.

The number of unknown signals is equal to the number of equation in the 
subsystem \eqref{eq:SEDAEa}, and \eqref{eq:SEDAEb}.
Such a quadratic system is a low-index DAE system if $A_{22}$ is invertible.
In this linear case \eqref{eq:SEDAEb} gives
\begin{equation*}
    {\bf x}_2=-A_{22}^{-1}(A_{21}{\bf x}_1+B_2{\bf u})
\end{equation*}
and substitution into \eqref{eq:SEDAEa} yields an ODE system.

In~\cite{IFACWC2017:Index:}, the concept of low index was extended to a structural property of non-quadratic, 
structurally overdetermined PSO systems, 
and model~\eqref{eq:SEDAE} has low structural index if and only if the matrix
\begin{equation*}
  \begin{bmatrix}
    A_{22}\\ C_{2}
  \end{bmatrix}
\end{equation*}
has full structural column rank. It was shown how such
structurally overdetermined, low-index submodels can be used for residual 
generation for more general non-linear DAE systems,
either by direct sequential residual generation with integral
causality or observer based residual methods. 

In Section~IV.B of~\cite{SMCA12:CausalIsol}, an operator was introduced that returns
the largest PSO set that is structurally monitorable under integral causality, and a
semi-explicit DAE has low structural index if and only if it is structurally
monitorable under integral causality; see Theorem~14 in~\cite{IFACWC2017:Index:}.
The submodel $M^*$ is defined as largest PSO subset of $M$ that is structurally
monitorable under integral causality.

Applying the algorithm in Section~\ref{sec:algorithm} with the operator
$M^{*}$ defined above yields four IRG sets, and the corresponding irreducible  fault signatures,
as shown in Table~\ref{tab:irg}.
\begin{table}[h]
\centering
\begin{tabular}{|l|l|}
\hline
IRG set & Fault signature \\
\hline
$\{e_{1},e_{2},e_{3}, e_{6}\}$ & $\{f_{1}\}$\\
$\{e_{1},e_{2},e_{3}, e_{4}, e_{6}\}$ & $\{f_{1}, f_{2}\}$\\
$\{e_{1},e_{2},e_{3}, e_{5}, e_{6}\}$ & $\{f_{1}, f_{3}\}$\\
$\{e_{4},e_{5},e_{6}\}$ & $\{f_{2}, f_{3}\}$\\
\hline
\end{tabular}
\caption{IRG sets and fault signatures of the model \eqref{eq:exDAE}}
\label{tab:irg}
\end{table}
 
Table~\ref{tab:mtes} lists the MTESs and their corresponding test supports for the model. 
The following observations can be made. First, $\{e_{1},e_{2},e_{3},e_{9}\}$ 
is the only IRG set that is also an MTES. This is because the MTESs 
$\{e_{1},e_{3},e_{4},e_{6}\}$ and $\{e_{1},e_{3},e_{5},e_{6}\}$ are not low-index DAEs, 
since in both cases the variable $x_2$ cannot be determined algebraically. However, 
by adding equation $e_2$ to these sets, the second and third IRG sets are obtained, 
at the cost of increased fault sensitivity due to the inclusion of the fault $f_1$.

Second, in the MTES framework, for each individual fault there exists a fault signature 
that contains this fault but not the others, whereas under the computational constraint 
the fault mode $\{f_{3}\}$ is not isolable from the mode $\{f_{1},f_{2}\}$.

\begin{table}[h]
\centering
\begin{tabular}{|l|l|}
\hline
MTES & Test support \\
\hline
$\{e_{1},e_{2},e_{3}, e_{6}\}$ & $\{f_{1}\}$\\
$\{e_{1},e_{3}, e_{4}, e_{6}\}$ & $\{f_{2}\}$\\
$\{e_{1},e_{3}, e_{5}, e_{6}\}$ & $\{f_{3}\}$\\
\hline
\end{tabular}
\caption{MTESs and test supports of the model \eqref{eq:exDAE}}
\label{tab:mtes}
\end{table}

  \section{Conclusions}

The paper has introduced a generalized structural framework for residual generation 
that simultaneously accounts for fault locations and constraints imposed by a chosen 
residual generation methodology.
  By extending the notions of PSO sets, MSO sets, and TES/MTES to testable PSO sets and RG sets, 
  and by characterizing irreducible fault signature sets as the join-irreducible elements of the resulting lattice, 
  the framework captures both fault isolability and computational limitations in a unified way.

An algorithm for computing all RG sets has been proposed, 
obtained by replacing the overdetermined-part operator with 
an operator $M^*$ that encodes the admissible residual generation method. 
The DAE example demonstrates how structural differential index 
can be used to define $M^*$ and to derive IRG sets and their associated irreducible fault signatures. 
 The results show that IRG sets play, for method-constrained 
diagnosis, the same central role that MTESs play in previous works.

\bibliographystyle{plain}
\bibliography{references}

\end{document}